\documentclass[submission,copyright,creativecommons]{eptcs}
\providecommand{\event}{TARK 2017} 
\usepackage{breakurl}             
\usepackage{underscore}           

\usepackage{enumitem}
\usepackage{comment}
\usepackage{xspace}
\usepackage{tikz}
\usepackage{url}
\usepackage{subfigure}
\usetikzlibrary{shapes,arrows,snakes,backgrounds,automata,patterns,trees}
\usetikzlibrary{circuits.logic.CDH}

\newcommand{\Agents}{{\mathit Ag}}

\newcommand{\putaway}[1]{}

\newcommand{\node}[2]{\langle #1 \circ #2 \rangle}

\newtheorem{theorem}{Theorem}
 \newtheorem{corollary}[theorem]{Corollary}
 \newtheorem{proposition}[theorem]{Proposition} 
\newtheorem{lemma}[theorem]{Lemma}
\newtheorem{example}{Example}
\newtheorem{definition}{Definition} 

 \newenvironment{proof}{\textbf{Proof.~ }}{\qed}
\newenvironment{proofsk}{\textbf{Sketch of Proof.\ }}{\qed}

\long\def\eatpar#1{%
\ifx#1\par                      
\let\nextmove=\eatpar           
\else
\let\nextmove=#1
\fi
\nextmove
}

\def\qed{\hfill{\qedboxempty}      
  \ifdim\lastskip<\medskipamount \removelastskip\penalty55\medskip\fi}

\def\qedboxempty{\vbox{\hrule\hbox{\vrule\kern3pt
                 \vbox{\kern3pt\kern3pt}\kern3pt\vrule}\hrule}}

\def\qedfull{\hfill{\qedboxfull}   
  \ifdim\lastskip<\medskipamount \removelastskip\penalty55\medskip\fi}

\def\qedboxfull{\vrule height 4pt width 4pt depth 0pt}

\newcommand{{\incolumn}}[1]{\begin{tabular}[c]{c} #1 \end{tabular}}
\newcommand{{\incolumnmath}}[1]{\begin{array}[c]{c} #1 \end{array}}

\begingroup
\catcode`\~=11
\gdef\urltilde{\lower 0.6ex\hbox{~}}
\endgroup

\newcommand{\E}{\mathcal{E}} \newcommand{\F}{\mathcal{F}}

\newcommand{\K}{\mathcal{K}} \renewcommand{\L}{\mathcal{L}}
\newcommand{\M}{\mathcal{M}} \newcommand{\N}{\mathcal{N}}

\newcommand{\defterm}[1]{\mbox{\underline{\it\smash{#1}\vphantom{\lower.1ex\hbox
{x}}}}}

\newcommand{\myi}{\emph{(i)}\xspace}
\newcommand{\myii}{\emph{(ii)}\xspace}
\newcommand{\myiii}{\emph{(iii)}\xspace}
\newcommand{\myiv}{\emph{(iv)}\xspace}
\newcommand{\myv}{\emph{(v)}\xspace}
\newcommand{\myvi}{\emph{(vi)}\xspace}

\newcommand{\Nat}{{\rm I\kern-.23em N}}

\newcommand{\commentout}[1]{}

\usepackage{marginnote}
\edef\marginnotetextwidth{\the\textwidth}

\newcommand{\CASE}[1]{\STATE \textbf{case} #1\textbf{:} \begin{ALC@g}}
\newcommand{\ENDCASE}{\end{ALC@g}}

\newcommand{\DEFAULT}{\STATE \textbf{default:} \begin{ALC@g}}
\newcommand{\ENDDEFAULT}{\end{ALC@g}}
\newcommand{\DEFAULTLINE}[1]{\STATE \textbf{default:} }

\usepackage{amssymb,amsmath}
\usepackage{graphicx}

\def\denot#1{[\![ #1 ]\!]}

\newcommand{\atleast}{\alpha}

\title{A Logic for Global and Local Announcements}
\author{Francesco Belardinelli
\institute{Labortoire IBISC, UEVE}
\institute{and IRIT Toulouse}
\email{belardinelli@ibsc.fr}
\and
Hans van Ditmarsch
       \institute{LORIA -- CNRS}
       \institute{Universit\'e de Lorraine}
       \institute{Vandoeuvre-l\`es-Nancy, France}
       \email{hvd@us.es}
\and 
Wiebe van der Hoek
       \institute{Department of Computing}
       \institute{University of Liverpool}
       \institute{Liverpool, UK}
       \email{wiebe@csc.liv.ac.uk}
}

\def\titlerunning{A Logic for Global and Local Announcements}
\def\authorrunning{F. Belardinelli, H. van Ditmarsch \& W. van der Hoek}

\begin{document}

\maketitle

\begin{abstract}
In this paper we introduce {\em global and local announcement logic}
(GLAL), a dynamic epistemic logic with two distinct announcement
operators -- $[\phi]^+_A$ and $[\phi]^-_A$ indexed to a subset $A$ of
the set $Ag$ of all agents -- for global and local
announcements respectively. The boundary case $[\phi]^+_{Ag}$
corresponds to the public announcement of $\phi$, as known from the
literature. Unlike standard public announcements, which are {\em
model transformers}, the global and local announcements are {\em
pointed model transformers}. In particular, the update induced by the
announcement may be different in different states of the
model. Therefore, the resulting computations are trees of models,
rather than the typical sequences. A consequence of our semantics is
that modally bisimilar states may be distinguished in our logic. Then, we
provide a stronger notion of bisimilarity and we show that it
preserves modal equivalence in GLAL.  Additionally, we show that GLAL
is strictly more expressive than public announcement logic with common
knowledge.
We prove a wide range of validities for GLAL
involving the interaction between dynamics and knowledge, and show
that the satisfiability problem for GLAL is decidable. We illustrate the formal
machinery by means of detailed epistemic scenarios.
\end{abstract}








\section{Introduction}

In this paper we take inspiration from the state of the art in public
announcement logic (PAL) and introduce a logic for global as well as
local announcements.
Public announcement logic has two
key features. First, announcements are {\em public}, in the sense that
all agents equally observe the new information, and are (commonly)
aware of all equally observing the information. Second,
announcements are {\em global}, that is, although for
truthful public announcements the truth of the announced formula in
the actual state is a precondition, how the new information is
processed does not depend on the actual state but rather on the model
(i.e., public announcements are model transformers).

In the proposed framework we carefully distinguish the two,
independent features of {\em publicity} and {\em globality}, which are
packed together in the announcement operator $[\phi]$, and relax both
of them.  Hence, by weakening publicity, we allow to make
announcements to a proper subset $A$ of the set $Ag$ of all
agents. Then, only the agents in $A$ partake of the new information
contained in the announcement. Further, by weakening globality, we
distinguish between local announcements, whose meaning depends on the
actual state, and global announcements that depend on general features
of the model.

As a result, the language of global and local announcement logic
(GLAL) contains two modalities $[\phi]^+_A$ and $[\phi]^-_A$, for the
global and local announcement of formula $\phi$ respectively, each of
them indexed to a coalition $A$ of agents. We endow GLAL with a
semantics in terms of {\em pointed model} updates that reflects the
intuitions illustrated above. Most interestingly, we are able to
provide a unified account of both global and local announcements, in
which the difference between the two depends on a subtle distinction
in the update mechanism.

{\bf Related Work.} Public announcement logics have witnessed a wealth
of contributions in recent
years \cite{Plaza89,pabe:2004,jfak.book:2011,hvdetal.handbook:2015},
thus making virtually impossible to give an exhaustive account of this
research area.  Here we mention the references most closely related
to the present endeavour, 
as well as some surveys on PAL.
In \cite{gerbrandyetal:1997} a logic of fully private announcements
 was proposed,
 while 
\cite{hvd.thesis:2000,hoekpauly:06a,baltagetal.hpi:2008,kooietal:2011} put forward logics
 of semi-private announcements, which relax the publicity assumption
 of PAL in various directions. Such private or semi-private
 announcements have also been modelled as action
 models \cite{baltagetal:1998}. Differently from our proposal, in
 semi-private announcements the agents that do not observe the
 announcement of $\phi$ learn at least that the other agents have
 learnt whether $\phi$. No such assumption holds in the present
 context. On the other hand, in fully private announcements the other
 agents learn nothing at all about the agents learning $\phi$ (which
 is typically interpreted as the other agents not even being aware of
 the announcement having taken place). This is also different from our
 setting, wherein these other agents learn something about $\phi$
 w.r.t.~the actual state.

Modal logics based on model transformations have also been proposed
in \cite{aucheretal:2009,jfak.sabotage:2005,fervari:2014,arecesetal:2012}. These
accounts share the aspect of locality (dependence of the model
transforming operation on the actual state) that also characterizes
our approach. However, differently from our proposal, these are very
expressive formalisms (typically undecidable, or non-axiomatizable)
that allow to add or remove individual pairs of states from an agent's
accessibility relation; thus operating on a purely semantic level.
On the contrary, in GLAL the model
transformation is determined by the announced formula, so that only
pairs satisfying a condition relative to this formula can be
removed. 
Our logic is therefore decidable, although we have not yet provided an axiomatization.
We also provide results on
bisimulations and the model checking problem. 



%

{\bf Schema of the paper.} In Section~\ref{preliminaries} we introduce
the syntax and semantics of GLAL and provide two examples to
illustrate the formalism.  In Section~\ref{comparison} we compare our
approach to relevant related account appearing in the literature on
dynamic epistemic logic (DEL); while in Sections~\ref{validities}
and \ref{expressivity} we analyse the expressivity of GLAL and prove that, 
differently from PAL, GLAL cannot be reduced to epistemic logic (with
common knowledge) as it is strictly more expressive.  In
Section~\ref{bisimulations} we introduce bisimulation relations that
preserve GLAL formulas; whereas in Section~\ref{modelcheck} we present
results on the model checking and satisfiability problems. We discuss
the meaning and relevance of these results in Section~\ref{conc},
where we also point to directions for future research.

\section{A Unified Framework for Global and Local Announcements} \label{preliminaries}


In this section we introduce the syntax and semantics of global and
local announcement logic.
%
We warn that the term `announcement' is used here with a different
meaning with respect to public announcement logic. As discussed in the
introduction, the announcements of PAL appear here as global
announcements to all agents. Hence, our notion of announcement is more
general as it also covers local announcements and announcement to only
a selected subset of all agents.  The distinction will be clear once
the appropriate semantics is introduced.\\


\textbf{Syntax.}
In the rest of the paper we assume a set $AP$ of atomic propositions
(or atoms), and a finite set $Ag$ of indexes for agents. 
%
\begin{definition}[GLAL] \label{pAL}
The formulas in GLAL are defined in BNF as follows, 
for $p \in AP$ and $A \subseteq Ag$:
\begin{eqnarray*}
\psi & ::= & p \mid \neg \psi \mid \psi \land \psi \mid C_A \psi \mid
     [\psi]^+_A \psi \mid [\psi]^-_A \psi
\end{eqnarray*}
\end{definition}

The language $\L_{glal}$ of GLAL contains epistemic formulas $C_A
\phi$, for coalition $A \subseteq Ag$ of agents, which intuitively
says that ``$\phi$ is common knowledge in coalition $A$'', as
customary. Moreover, we have global announcement formulas
$[\psi]^+_A \phi$, whose reading is that ``after globally announcing
$\psi$ to the agents in $A$, $\phi$ is true'', as well as local
announcements $[\psi]^-_A \phi$, whose meaning is that
``after locally announcing $\psi$ to the agents in $A$, $\phi$ is
true''.  We will illustrate and discuss, using our semantics, the
different interpretations of operators $[\psi]^+_A$ and $[\psi]^-_A$.
The individual knowledge formulas $K_a \phi$ can be defined as $C_{\{
a \}} \phi$ as standard, as well as symbols $\bot$, $\top$,
connectives $\lor$, $\to$, and dual operators $M_a$, $\langle \psi
\rangle^-_A$, and $\langle \psi \rangle^+_A$.  Also, the ``everybody
knows'' formula $E_A \phi$ is a shorthand for $\bigwedge_{a \in A} K_a
\phi$. 
We omit subscript $A$ from $E_A \phi$ and $C_A \phi$ whenever $A$ is
the grand coalition $Ag$, and simply write $E \phi$, $C
\phi$. Finally, we write $Kw_a \phi$ as a shorthand for $K_a \phi \lor K_a \neg \phi$.

Global and local announcement logic extends a number of well-known
formalisms.
The language $\L_{pal^+}$ obtained by Def.~\ref{pAL} without clause
$[\psi]^-_A \psi$ is (an extension of) public announcement logic; the
language $\L_{el}$ without clause $[\psi]^+_A \psi$ as well is
epistemic logic (with common knowledge), and language $\L_{pl}$
without clause $C_A \psi$ as well is propositional logic.
These (syntactic) language inclusions can be summarised as follows:
\begin{equation*}
\L_{pl}  \ \ \subseteq   \ \  \L_{el}   \ \ \subseteq   \ \ \L_{pal^+}   \ \ \subseteq   \ \  \L_{glal}
\end{equation*}

Hereafter, when we refer to ``epistemic logic'', we mean language
$\L_{el}$ including common knowledge.\\

\textbf{Semantics.} \label{semantics}
We interpret formulas in GLAL on multi-modal Kripke frames and models,
and then define appropriate update mechanisms for global and local
announcements.
\begin{definition}[Frame]\label{def:kf}
A Kripke frame is a tuple $\F = \langle W,  \{R_a\}_{a \in Ag} \rangle$ where
\begin{itemize}
\item $W$ is a set of {\em possible worlds};
\item for
  every agent index $a \in Ag$, $R_a \subseteq 2^{W \times W}$ is an {\em
    equivalence relation} on $W$.
\end{itemize}
\end{definition}

As customary in epistemic logic \cite{HoekM95,BlackburndRV01}, for
every agent $a \in Ag$, $R_a$ is the corresponding {\em
indistinguishability} relation between worlds in $W$.  In the
following, for a coalition $A
\subseteq Ag$, we consider also the
reflexive and transitive closure $R^C_A = (\bigcup_{a \in A} R_a)^*$
of the union of indistinguishability relations, for the interpretation
of common knowledge.  Then, for coalition $A$ and $w \in
W$, we set
$R^C_A(w) = \{w' \mid R^C_A(w,w') \}$.  Notice that, each $R_a$ being
an equivalence relation, each 
$R^C_A(w)$ is the equivalence class of $w \in W$.
Moreover, whenever $A$ is a singleton $\{a\}$, we
obtain that $R^C_A(w) =  \{w' \mid R_{\{a\}}(w,w') \} = R_a(w)$, and $R^C_A$ can be
represented as
the set $\E_A = \{ R^C_A(w) \mid w \in W\}$ of its equivalence classes.

To assign a meaning to formulas in GLAL we introduce {\em assignments}
as functions $V: AP \to 2^{W}$.
A {\em (Kripke) model} is then defined as a pair $\M = \langle
\F, V \rangle$.  
%
%
\begin{definition}[Satisfaction] \label{def:semantics}
We inductively define the {\em satisfaction set}
$[\![\varphi]\!]_{\M} \subseteq W$ of formula $\varphi$ in model $\M
= \langle \F, V \rangle$
as
follows:

{\small
\begin{tabbing}
$[\![p]\!]_{\M}$ \ \ \ \ \ \ \   \= =  \ \ \ \= $V(p)$\\
$[\![\neg \psi]\!]_{\M}$ \> = \> $W \setminus [\![\psi]\!]_{\M}$\\
$[\![\psi \land \psi']\!]_{\M}$ \> = \> $[\![\psi]\!]_{\M} \cap [\![\psi']\!]_{\M}$\\
$[\![C_A \psi]\!]_{\M}$ \> = \> $\{ w \in W \mid \text{for all }  w' \in R^C_A(w),  w' \in [\![\psi]\!]_{\M} \}$\\
$[\![[\psi]^-_A \psi']\!]_{\M}$ \> = \> $\{ w \in W \mid \text{if }  w \in [\![\psi]\!]_{\M} \text{ then }  w \in [\![\psi']\!]_{\M^-_{(w,\psi,A)}} \}$\\
$[\![[\psi]^+_A \psi']\!]_{\M}$ \> = \> $\{ w \in W \mid \text{if }  w \in [\![\psi]\!]_{\M} \text{ then }  w \in [\![\psi']\!]_{\M^+_{(w,\psi,A)}} \}$
\end{tabbing}
}
where {\em refinements} $\M^-_{(w,\psi,A)} = \langle W^-, \{ R^-_a
\}_{a \in Ag}, V^- \rangle$ and  $\M^+_{(w,\psi,A)} = \langle W^+, \{ R^+_a
\}_{a \in Ag}, V^+ \rangle$ of model $\M$ w.r.t.~world $w$, formula
$\psi$, and coalition $A$, are defined as
\begin{itemize}
\item $W^- = W^+ = W$ and $V^- = V^+ = V$; 
\item for every agent $b \notin A$, $R^-_b = R^+_b = R_b$;
  while for $a \in A$,
\[
R^-_a(v) =  
  \begin{cases} R_a(v) \cap [\![\psi]\!]_{\M} & \mbox{if } v \in R_a(w) \cap [\![\psi]\!]_{\M}\\
R_a(v) \cap [\![\neg \psi]\!]_{\M} & \mbox{if } v \in R_a(w) \cap [\![\neg \psi]\!]_{\M}\\
  R_a(v) & \mbox{otherwise}
  \end{cases}
  \]
\[
R^+_a(v) =  
  \begin{cases} R_a(v) \cap [\![\psi]\!]_{\M} & \mbox{if } v \in R^C_A(w) \cap [\![\psi]\!]_{\M}\\
R_a(v) \cap [\![\neg \psi]\!]_{\M} & \mbox{if } v \in R^C_A(w) \cap [\![\neg \psi]\!]_{\M}\\
  R_a(v) & \mbox{otherwise}
  \end{cases}
  \]
\end{itemize}
\end{definition}

By Def.~\ref{def:semantics} the refinement $\M^-_{(w,\psi,A)}$ only
affects worlds that are accessible by each agent in $A$ separately,
while $\M^+_{(w,\psi,A)}$ involves all worlds reachable through
relation $R^C_A$.  In all these worlds the accessibility relation is
updated according to whether the world in question satisfies the
announcement, that is, the announcement refines the equivalence class
of each such world.  In Example~\ref{ex1} and \ref{ex:bisim} below we
illustrate the differences between the two types of refinement.
Notice that in the case of single agents, the refinements
$\M^-_{(w,\psi,a)}$ and $\M^+_{(w,\psi,a)}$ coincide, hence we omit
superscripts $-$ and $+$ from single-agent refinements and modalities.
Indeed, globally or locally announcing a fact to a single agent is
tantamount, as she is the only one to witness the announcement.
In such a case, model refinement
$\M_{(w,\psi,a)}$
can be interpreted as ``in $R_a(w)$ agent $a$
learns whether $\psi$''.  As a consequence, formula $[\psi]_a
\phi$ 
then becomes: if $\psi$ holds and $a$ learns whether $\psi$, then
$\phi$ holds as well.  Also, the updated set
$\E'_a$ of equivalence classes in $\M_{(w,\psi,a)}$ can be shown to be equal to $(\E_a
\setminus \{R_a(w) \}) \cup \{R_a(w) \cap \denot{\psi}, R_a(w) \cap
\denot{\neg\psi}\}$.


We introduce standard notions of truth and validity.
A formula $\phi$ is {\em satisfied} at $w$, or $(\M, w) \models
\phi$, iff $w \in \denot{\phi}_{\M}$;  
$\phi$ is {\em true} at $w$, or $(\F, w) \models
\phi$, iff $(\langle \F, V \rangle, w) \models \phi$ for every
assignment $V$; $\phi$ is {\em valid} in a frame $\F$, or $\F \models
\phi$, iff $(\F, w) \models \phi$ for every world $w$ in $\F$.
We often omit the subscript
$\M$ from the satisfaction set $[\![\psi]\!]_{\M}$ whenever clear by
the context.

We now state that our model refinements are well-defined as both $R^-_{a}$
and $R^+_{a}$ are actually equivalence relations.
\begin{proposition} \label{prop11}
Let $\M$ be a model with refinements $\M^-_{(w,\psi,A)}$ and
$\M^+_{(w,\psi,A)}$. For every agent $a \in Ag$, if $R_a$ is an
equivalence relation, then also $R^-_{a}$ and $R^+_{a}$ are.
\end{proposition} 

We refer to Section~\ref{comparison} for a thorough comparison with related
approaches to public and private announcements in DEL.  Here, we 
illustrate the difference between global and local announcements by
means of two examples drawn from the literature on multi-agent systems
and dynamic logic \cite{DitmarschvdHK07,fhmv:rak}.
Hereafter we often represent a set as some
sequence of its elements.

\begin{example} \label{ex1}
Here we consider the well-known puzzle of the muddy children.  We
assume familiarity with this scenario and refer to
\cite{DitmarschK15,fhmv:rak} for a detailed presentation.  The initial
model $\M$ for 3 children (red, blue, and green), where no child knows
whether she is muddy, can be represented as follows:
\begin{center}
{\scriptsize  \begin{tikzpicture}[scale=1.4, rotate = 45]
 \tikzstyle{vertex}=[circle, minimum size=20pt, inner sep=0pt]
 \tikzstyle{selected vertex} = [vertex, fill=red!24]
 \tikzstyle{selected edge} = [draw,line width=5pt,-,red!50]
 \tikzstyle{edge} = [draw,thick,-,black]

 \node[vertex] (v8) at (-0.5,-0.5) {(\textcolor{red}{0}, \textcolor{green}{0}, \textcolor{blue}{0})};
 \node[vertex] (v9) at (-0.5, 1) {(\textcolor{red}{1}, \textcolor{green}{0}, \textcolor{blue}{0})};
 \node[vertex] (v13) at (-0.08,1.6) {(\textcolor{red}{1},\textcolor{green}{1},\textcolor{blue}{0})};
 \node[vertex] (v12) at (-0.08, 0.1) {(\textcolor{red}{0},\textcolor{green}{1},\textcolor{blue}{0})};
 \node[vertex] (v10) at (1,-0.5) {(\textcolor{red}{0},\textcolor{green}{0},\textcolor{blue}{1})};
 \node[vertex] (v14) at (1.42,0.1) {(\textcolor{red}{0},\textcolor{green}{1},\textcolor{blue}{1})};
 \node[vertex] (v11) at (1,1) {(\textcolor{red}{1},\textcolor{green}{0},\textcolor{blue}{1})};
 \node[vertex] (v15) at (1.42, 1.6) {(\textcolor{red}{1},\textcolor{green}{1},\textcolor{blue}{1})};

 \path (v8) edge[color = red] node[below left] {$r$} (v9); 
 \path (v13) edge[color = red] node[below] {$r$} (v12); 
 \path (v10) edge[color = red] node[above] {$r$} (v11); 
 \path (v15) edge[color = red] node[above right] {$r$} (v14); 
 \path (v8) edge[color = blue] node[below right] {$b$} (v10); 
 \path (v14) edge[color = blue] node[below right] {$b$} (v12); 
 \path (v9) edge[color = blue] node[above left] {$b$} (v11); 
 \path (v15) edge[color = blue] node[above left] {$b$} (v13); 
 \path (v14) edge[color = green] node[right] {$g$} (v10); 
 \path (v8) edge[color = green] node[right] {$g$} (v12); 
 \path (v15) edge[color = green] node[left] {$g$} (v11); 
 \path (v9) edge[color = green] node[left] {$g$} (v13); 
  \end{tikzpicture}}
\end{center}

Now suppose that only red is muddy, i.e., the actual world is
$(1,0,0)$, and the father {\em locally} announces to red, green, and
blue that at least one child is muddy, that is, he announces that
formula $\atleast := m_r
\lor m_b \lor m_g$ is true.  The updated model $\M^-_{(100, \atleast, rgb)}$ is then given as follows, on the left:
\begin{center}
\begin{tabular}{cc}
{\scriptsize  \begin{tikzpicture}[scale=1.3, rotate = 45]
 \tikzstyle{vertex}=[circle, minimum size=20pt, inner sep=0pt]
 \tikzstyle{selected vertex} = [vertex, fill=red!24]
 \tikzstyle{selected edge} = [draw,line width=5pt,-,red!50]
 \tikzstyle{edge} = [draw,thick,-,black]

 \node[vertex] (v8) at (-0.5,-0.5) {(\textcolor{red}{0}, \textcolor{green}{0}, \textcolor{blue}{0})};
 \node[vertex] (v9) at (-0.5, 1) {(\textcolor{red}{1}, \textcolor{green}{0}, \textcolor{blue}{0})};
 \node[vertex] (v13) at (-0.08,1.6) {(\textcolor{red}{1},\textcolor{green}{1},\textcolor{blue}{0})};
 \node[vertex] (v12) at (-0.08, 0.1) {(\textcolor{red}{0},\textcolor{green}{1},\textcolor{blue}{0})};
 \node[vertex] (v10) at (1,-0.5) {(\textcolor{red}{0},\textcolor{green}{0},\textcolor{blue}{1})};
 \node[vertex] (v14) at (1.42,0.1) {(\textcolor{red}{0},\textcolor{green}{1},\textcolor{blue}{1})};
 \node[vertex] (v11) at (1,1) {(\textcolor{red}{1},\textcolor{green}{0},\textcolor{blue}{1})};
 \node[vertex] (v15) at (1.42, 1.6) {(\textcolor{red}{1},\textcolor{green}{1},\textcolor{blue}{1})};

 \path (v13) edge[color = red] node[below] {$r$} (v12); 
 \path (v10) edge[color = red] node[above] {$r$} (v11); 
 \path (v15) edge[color = red] node[above right] {$r$} (v14); 
 \path (v8) edge[color = blue] node[below right] {$b$} (v10); 
 \path (v14) edge[color = blue] node[below right] {$b$} (v12); 
 \path (v9) edge[color = blue] node[above left] {$b$} (v11); 
 \path (v15) edge[color = blue] node[above left] {$b$} (v13); 
 \path (v14) edge[color = green] node[right] {$g$} (v10); 
 \path (v8) edge[color = green] node[right] {$g$} (v12); 
 \path (v15) edge[color = green] node[left] {$g$} (v11); 
 \path (v9) edge[color = green] node[left] {$g$} (v13); 

  \end{tikzpicture} } &  {\scriptsize  \begin{tikzpicture}[scale=1.3, rotate = 45]
 \tikzstyle{vertex}=[circle, minimum size=20pt, inner sep=0pt]
 \tikzstyle{selected vertex} = [vertex, fill=red!24]
 \tikzstyle{selected edge} = [draw,line width=5pt,-,red!50]
 \tikzstyle{edge} = [draw,thick,-,black]

 \node[vertex] (v8) at (-0.5,-0.5) {(\textcolor{red}{0}, \textcolor{green}{0}, \textcolor{blue}{0})};
 \node[vertex] (v9) at (-0.5, 1) {(\textcolor{red}{1}, \textcolor{green}{0}, \textcolor{blue}{0})};
 \node[vertex] (v13) at (-0.08,1.6) {(\textcolor{red}{1},\textcolor{green}{1},\textcolor{blue}{0})};
 \node[vertex] (v12) at (-0.08, 0.1) {(\textcolor{red}{0},\textcolor{green}{1},\textcolor{blue}{0})};
 \node[vertex] (v10) at (1,-0.5) {(\textcolor{red}{0},\textcolor{green}{0},\textcolor{blue}{1})};
 \node[vertex] (v14) at (1.42,0.1) {(\textcolor{red}{0},\textcolor{green}{1},\textcolor{blue}{1})};
 \node[vertex] (v11) at (1,1) {(\textcolor{red}{1},\textcolor{green}{0},\textcolor{blue}{1})};
 \node[vertex] (v15) at (1.42, 1.6) {(\textcolor{red}{1},\textcolor{green}{1},\textcolor{blue}{1})};

 \path (v13) edge[color = red] node[below] {$r$} (v12); 
 \path (v10) edge[color = red] node[above] {$r$} (v11); 
 \path (v15) edge[color = red] node[above right] {$r$} (v14); 
 \path (v14) edge[color = blue] node[below right] {$b$} (v12); 
 \path (v9) edge[color = blue] node[above left] {$b$} (v11); 
 \path (v15) edge[color = blue] node[above left] {$b$} (v13); 
 \path (v14) edge[color = green] node[right] {$g$} (v10); 
 \path (v8) edge[color = green] node[right] {$g$} (v12); 
 \path (v15) edge[color = green] node[left] {$g$} (v11); 
 \path (v9) edge[color = green] node[left] {$g$} (v13); 

  \end{tikzpicture} } \\
\end{tabular}
\end{center}

Notice that only the indistinguishability relation for red is updated,
as in all worlds that blue and green consider possible from $(1,0,0)$,
formula $\atleast$ is indeed true.
Hence, after the father's local announcement, in $(1,0,0)$ all
children know that at least one child is muddy, i.e., $(\M, (1,0,0)) \models
[\atleast]^-_{rgb} E \atleast$. Moreover, red learns that she is
muddy, i.e., $(\M, (1,0,0)) \models [\atleast]^-_{rgb} K_r m_r$.  

On the other hand, the father's local announcement is not enough to
make $\atleast$ common knowledge for red, green and blue, that is,
$(\M, (1,0,0)) \not \models [\atleast]^-_{rgb} C_{rgb} \atleast$ as, for
instance, blue considers possible that red considers possible that
blue considers possible that no child is muddy, that is,
$(\M, (1,0,0)) \models [\atleast]^-_{rgb} M_b M_r M_b \neg \atleast$ via
epistemic path $(1,0,0)
\sim_b (1,1,0) \sim_r (0,1,0) \sim_b (0,0,0)$. This is in contrast
with the classic version of the muddy children puzzle with public
announcements. In general, for every state $s \in \{0, 1 \}^3$
different from $(0,0,0)$, announcing privately $\atleast$ in $s$ is
not sufficient to derive common knowledge of $\atleast$:
$(\M, s) \not \models [\atleast]^-_{rgb} C_{rgb} \atleast$

Now suppose that at the beginning, again in world $(1,0,0)$, the
father {\em globally} announces to red and blue only that at least one
child is muddy. The updated model $\M^+_{(100,\atleast,rb)}$ is 
shown above on the right.
%
%
Specifically, in $\M^+_{(100,\atleast,rb)}$ the indistinguishability
relations for both red and blue are updated, and as a result, after
the father's global announcement, in $(1,0,0)$ red and blue have
common knowledge that at least one child is muddy:
$(\M, (1,0,0)) \models [\atleast]^+_{rb} C_{rb} \atleast$.  
%
However, also in this case the father's global announcement is not
enough to make $\atleast$ common knowledge amongst all children, that
is, $(\M, (1,0,0)) \not \models [\atleast]^+_{rb} C_{rgb} \atleast$.


\end{example}

\begin{example} \label{ex:bisim}
Here we consider a simple scenario of communication between a sender
$s$ and a receiver $r$ over a reliable channel that is listened to by
eavesdropper $e$.  The initial state is represented by the following
model $\N$, in which $s$ knows the actual value of the bit (either 0
or 1), while $r$ and $e$ are unsure about it.
\begin{center}
  \begin{tikzpicture}[scale=1.3, node distance = 2cm]
 \tikzstyle{vertex}=[circle, draw, minimum size=20pt, inner sep=0pt]
 \tikzstyle{selected vertex} = [vertex, fill=red!24]
 \tikzstyle{selected edge} = [draw,line width=5pt,-,red!50]
 \tikzstyle{edge} = [draw,thick,-,black]

 \node[vertex, label=below:{$w_1$}] (v1) {$0$};
 \node[vertex, label=below:{$w_2$}] (v2) [right of = v1]  {$1$};

 \path (v1) edge[<->] node[below] {$r,e$} (v2); 
 \path (v1) edge[loop left] node[left] {$s, r,e$} (v1); 
 \path (v2) edge[loop right] node[right] {$s, r,e$} (v2); 
  \end{tikzpicture}
\end{center}

Then, after $s$ has communicated to $r$ the value of the bit, we
obtain the updated model $\N_{(w_1, bit = 0, r)}$:
\begin{center}
  \begin{tikzpicture}[scale=1.3, node distance = 2cm]
 \tikzstyle{vertex}=[circle, draw, minimum size=20pt, inner sep=0pt]
 \tikzstyle{selected vertex} = [vertex, fill=red!24]
 \tikzstyle{selected edge} = [draw,line width=5pt,-,red!50]
 \tikzstyle{edge} = [draw,thick,-,black]

 \node[vertex, label=below:{$w_1$}] (v1) {$0$};
 \node[vertex, label=below:{$w_2$}] (v2) [right of = v1]  {$1$};

 \path (v1) edge[<->] node[below] {$e$} (v2); 
 \path (v1) edge[loop left] node[left] {$s, r,e$} (v1); 
 \path (v2) edge[loop right] node[right] {$s, r,e$} (v2); 
  \end{tikzpicture}
\end{center}

Hence, we have that $(\N, w_1) \models [bit = 0]_{r} K_r (bit =
0)$, as expected. On the other hand, the eavesdropper does not learn
the value of the bit, but she learns that $r$ knows it: $(\N,
w_1) \models [bit = 0]_{r} (\neg Kw_e (bit = 0) \land K_e Kw_r (bit =
0))$.

Now compare model $\N$ above with the following model $\N'$,
which is bisimilar in the standard sense \cite{BlackburndRV01}
(which we call modally bisimilar or $m$-bisimilar):
\begin{center}
  \begin{tikzpicture}[scale=1.3, node distance = 2cm]
 \tikzstyle{vertex}=[circle, draw, minimum size=20pt, inner sep=0pt]
 \tikzstyle{selected vertex} = [vertex, fill=red!24]
 \tikzstyle{selected edge} = [draw,line width=5pt,-,red!50]
 \tikzstyle{edge} = [draw,thick,-,black]

 \node[vertex, label=above:{$v'_1$}] (v1) {$0$};
 \node[vertex, label=above:{$v'_2$}] (v2) [right of = v1]  {$1$};
 \node[vertex, label=below:{$w'_1$}, node distance = 1.8cm] (v3) [below of = v1]  {$0$};
 \node[vertex, label=below:{$w'_2$}] (v4) [right of = v3]  {$1$};

 \path (v1) edge[<->] node[below] {$r,e$} (v2); 
 \path (v1) edge[loop left] node[left] {$s, r,e$} (v1); 
 \path (v2) edge[loop right] node[right] {$s, r,e$} (v2); 
 \path (v1) edge[<->] node[left] {$s, e$} (v3); 
 \path (v2) edge[<->] node[right] {$s, e$} (v4); 
 \path (v3) edge[<->] node[below] {$r,e$} (v4); 
 \path (v3) edge[loop left] node[left] {$s, r,e$} (v3); 
 \path (v4) edge[loop right] node[right] {$s, r,e$} (v4); 

  \end{tikzpicture}
\end{center}

However, this time, after communicating to $r$ the value of the bit,
we obtain the updated model $\N'_{(w'_1, bit = 0, r)}$, which is not
bisimilar to $\N'$:
\begin{center}
  \begin{tikzpicture}[scale=1.3, node distance = 2cm]
 \tikzstyle{vertex}=[circle, draw, minimum size=20pt, inner sep=0pt]
 \tikzstyle{selected vertex} = [vertex, fill=red!24]
 \tikzstyle{selected edge} = [draw,line width=5pt,-,red!50]
 \tikzstyle{edge} = [draw,thick,-,black]

 \node[vertex, label=above:{$v'_1$}] (v1) {$0$};
 \node[vertex, label=above:{$v'_2$}] (v2) [right of = v1]  {$1$};
 \node[vertex, label=below:{$w'_1$}, node distance = 1.8cm] (v3) [below of = v1]  {$0$};
 \node[vertex, label=below:{$w'_2$}] (v4) [right of = v3]  {$1$};

 \path (v1) edge[<->] node[below] {$r, e$} (v2); 
 \path (v1) edge[loop left] node[left] {$s, r,e$} (v1); 
 \path (v2) edge[loop right] node[right] {$s, r,e$} (v2); 
 \path (v1) edge[<->] node[left] {$s,e$} (v3); 
 \path (v2) edge[<->] node[right] {$s,e$} (v4); 
 \path (v3) edge[<->] node[below] {$e$} (v4); 
 \path (v3) edge[loop left] node[left] {$s, r,e$} (v3); 
 \path (v4) edge[loop right] node[right] {$s, r,e$} (v4); 

  \end{tikzpicture}
\end{center}

In particular, in $w'_1$ eavesdropper $e$ does not know that receiver
$r$ has learnt the value of the bit: $(\N', w'_1) \models [bit =
0]_{r} \neg K_e Kw_r (bit = 0)$.  We elaborate on the fact that
formulas in GLAL are not preserved under standard modal bisimulations.
Specifically, in model $\N'$, and differently from $\N$, sender
$s$ and eavesdropper $e$ are not able to distinguish between bisimilar
states $w'_1$ and $v'_1$, which are nonetheless distinct for receiver
$r$. We can interpret this feature of $\N'$ as formalising the fact
that $s$ and $e$ are uncertain as to $r$'s behaviour. Indeed, since
these states are indistinguishable for $s$ and $e$, then they must
differ as to the epistemic state of $r$. And this is reflected in the
different results of announcing $bit = 0$ to $r$.  This subtle
distinction is reminiscent of the notion of {\em attentive
announcements} put forward in \cite{BolanderDHLPS16}. We discuss this
point in detail in Section~\ref{comparison}.
\end{example}

These examples are intended to illustrate the formal
features of GLAL to represent global and local communication. 
In particular, GLAL allows to express local communication that cannot
be captured in PAL. In Section~\ref{expressivity} we analyse the
expressivity of GLAL and provide a formal proof of the fact that it is
strictly more expressive than PAL. But first we compare GLAL to
related accounts in the literature.

\section{Discussion and Comparison} \label{comparison}


In this section we compare our logic to other dynamic epistemic
logics, and to accounts of awareness.

\paragraph*{Public announcements}

The logic GLAL can embed public announcement
 logic \cite{baltagetal:1998}.  We show that the global announcement
 modality $[\phi]^+_{Ag}$ for the grand coalition simulates operator
 $[\phi]$ from PAL. Let us recall the satisfaction clause for
 $[\phi]$-formulas in PAL:
\begin{eqnarray*}
[\![[\psi] \psi']\!]_{\M} & = & \{ w \in W \mid \text{if }  w \in [\![\psi]\!]_{\M} \text{ then }  w \in [\![\psi']\!]_{\M_{\psi}} \}
\end{eqnarray*}
where the {\em refinement} $\M_{\psi} = \langle W_{\psi}, \{ R_{\psi,a}
\}_{a \in Ag}, V_{\psi} \rangle$ of model $\M$ w.r.t.~formula
$\psi$ is defined as
\myi $W_{\psi} = W \cap [\![\psi]\!]_{\M}$; 
\myii for every agent $a \in Ag$, $R_{\psi,a} = R_a \cap ([\![\psi]\!]_{\M} \times [\![\psi]\!]_{\M})$;
and \myiii for every $p \in AP$, $V_{\psi}(p) = V(p) \cap [\![\psi]\!]_{\M}$.
%
Intuitively, $\M_{\psi}$ is the restriction of $\M$ to the worlds
satisfying $\psi$.

Now consider the recursively defined embedding $tr: \L_{pal} \to \L_{glal}$ with only non-trivial clause
$tr([\phi] \phi') = [tr(\phi)]^+_{Ag} tr(\phi')$. It is then easy to prove that
\begin{proposition} \label{prop10}
    For all formulas $\psi$ in PAL,
      $(\M, w) \models \psi$ iff $(\M, w) \models tr(\psi)$.
\end{proposition}
\begin{proofsk}
The only non-trivial case is for $\psi = [\phi] \phi'$. This follows
 by inductive hypothesis on $\phi$ and $\phi'$, by observing that in
 $\M^+_{(w, \phi, Ag)}$, any state $w'$ is reachable from $w$ via
 $R^{+C}_{Ag}$ iff $w'$ is reachable from $w$ in $\M_{\phi}$.
\end{proofsk}

The following corollary follows immediately from Proposition~\ref{prop10}.
\begin{corollary} \label{prop1}
    For all formulas $\phi, \psi$ in PAL,
     $(\M, w) \models [\phi]^+_{Ag} \psi$ iff $(\M, w) \models [\phi] \psi$
\end{corollary}
In Section \ref{expressivity} we use that public
announcement logic can be embedded into GLAL, to demonstrate that GLAL
is at least as expressive as PAL.

\paragraph*{Private announcements}
A local announcement is a private announcement to some agents. It will
therefore come as no surprise that there is also a strong relation
between GLAL and private announcements. Semi-private announcements
have been proposed and discussed in
in \cite{gerbrandyetal:1997,hvd.thesis:2000,hoekpauly:06a,baltagetal.hpi:2008}.
Specifically, after announcing semi-privately $\phi$ to coalition $A$,
all agents in $A$ know that $\phi$ is true, and the agents in
$Ag \setminus A$ know that all agents in $A$ know whether $\phi$ is
true. Now compare this to a local announcement of $\phi$ to coalition
$A$, after which all agents in $A$ know that $\phi$ is true, and the
agents in $Ag \setminus A$ are uncertain between all agents in $A$
knowing that $\phi$ is true, or not having observed the
announcement. The distinction between knowing whether and knowing that
is not so fundamental here (we could just as well have tweaked the
semantics to have knowing whether announcements $[\phi]^+_A$ and
$[\phi]^-_A$.) 
Whereas the
locality of our framework is an essential difference.

The standard way to define the semantics of semi-private announcement
is by refinement of the accessibility relation, namely as $R^{sp}_a =
R_a$ for $a \not\in A$, whereas $R^{sp}_a = R_a \cap
([\![\psi]\!]_{\M}^2 \cup [\![{\neg\psi}]\!]_{\M}^2))$ for $a \in
A$. Staying close to our semantics, and using, as for public
announcements, the isomorphy of point-generated submodels,
semi-private announcements can be interpreted by model refinement
$\M^{sp}_{(w,\psi,a)}$ according to which $W^{sp} = W$, $V^{sp} = V$,
and for $a \in A$,
\[
R^{sp}_a(v) =  
  \begin{cases} R_a(v) \cap [\![\psi]\!]_{\M} & \mbox{if } v \in R^C_{Ag}(w) \cap [\![\psi]\!]_{\M}\\
R_a(v) \cap [\![\neg \psi]\!]_{\M} & \mbox{if } v \in R^C_{Ag}(w) \cap [\![{\neg \psi}]\!]_{\M}\\
  R_a(v) & \mbox{otherwise} 
  \end{cases}
  \]
However, there is no embedding in this case.

\paragraph*{Attentive announcements}

In \cite{BolanderDHLPS16} the authors introduce a logic of {\em
attention-based announcements}: agents will only process the new
information $\phi$ if they are paying attention. Whether they pay
attention is handled by a designated set of propositional
variables. There is a surprising close relation to our logic GLAL,
despite the absence of such attention variables. Consider again
Example \ref{ex:bisim}, wherein we modelled the announcement $[bit
=0]$ to agent $r$ in state $w'_1$ of model $\N'$. Ignore the role of
agent $e$ in the modelling. Although $r$ processes the new
information, agent $s$ is uncertain about this fact. Agent $s$
considers that possible, but also considers it possible that $r$
remains uncertain between $bit=0$ and $bit=1$, namely in states $v_1'$
and $v_2'$. Now consider adding an `attention variable' $h_r$ for
agent $r$ to the model, as in \cite{BolanderDHLPS16}, such that $h_r$
is true in $w'_1$ and $w'_2$ but false in $v'_1$ and $v'_2$. Then
(modulo technical details) such an announcement of $bit=0$ to $r$
corresponds to an attention-based announcement of $bit=0$ to $r$
wherein agent $s$ is uncertain whether $r$ is paying attention. The
`technical details' in which the semantics differ are that (i)
in \cite{BolanderDHLPS16} truly private aspects of announcements are
modelled (the believed announcements of \cite{gerbrandyetal:1997}),
which do not preserve equivalence relations, whereas in our proposal
we consider semi-private announcements that preserve equivalence
relations; and (ii) our announcements are not necessarily public. But,
essentially, an announcement $[\phi]^-_A$ is strictly related to an
attention-based announcement of $\phi$ to the agents in $A$, wherein
agents not in $A$ are uncertain whether those in $A$ are {\em
individually} paying attention. On the other hand, an announcement
$[\phi]^+_A$ is related to an attention-based announcement of $\phi$
to the agents in $A$, wherein agents not in $A$ are uncertain whether
those in $A$ are {\em jointly} paying attention. It seems remarkable
that in GLAL we can simulate attention of agents without designated
 variables.



\section{Validities} \label{validities}

In this section we consider notable validities in GLAL, which shed
light on the meaning of the global and local announcement operators.
Firstly, we observe that
after announcing (truthfully) a propositional formula $\phi$ to the
agents in $A$, they know $\phi$.
\begin{proposition} \label{lem1}
For every propositional formula $\phi \in \L_{pl}$,
\begin{eqnarray}
 \models [\phi]^-_A E_A \phi \label{ref1}\\
 \models [\phi]^+_A C_A \phi \label{ref2}
\end{eqnarray}
\end{proposition}
%

According to Proposition~\ref{lem1}, if a propositional formula $\phi$
is announced locally, then all agents involved in the announcement
jointly know $\phi$; while if $\phi$ is announced globally, then it also
becomes common knowledge.
%
 Proposition~\ref{lem1} does not hold for general formulas
$\phi \in \L_{glal}$. As a counterexample, consider Moore's formula
$p \land \neg K_a p$.

We anticipate that, as a consequence of Theorem~\ref{theor1} below,
differently from PAL there is no set of validities in GLAL to rewrite
any announcement in terms of simpler formulas.  Nonetheless,
the following formulas are validities in GLAL:
\begin{eqnarray*}
        [\phi]^-_A p &  \leftrightarrow & \phi \to p \\
        \left[\phi \right]^-_A  \neg \psi &   \leftrightarrow & \phi \to \neg [\phi]^-_A  \psi \\
        \left[\phi\right]^-_A  (\psi \land \psi') &  \leftrightarrow &  [\phi]^-_A  \psi \land [\phi]^-_A  \psi' 
\end{eqnarray*}
and similarly for $\left[\phi\right]^+_A$.
%
Thus, both announcement operators commute with propositional
connective.

Moreover, 
epistemic operators and nested announcements commute with the announcement operators if they
 refer to the 
same coalition of agents.
\begin{proposition} \label{prop20}
The following are GLAL validities:
\begin{eqnarray} 
      \left[\phi \right]^-_A  E_A \psi  & \leftrightarrow & \phi \to E_A [\phi]^-_A  \psi \label{eq1}\\
\left[\phi \right]^-_A \left[\phi' \right]^-_A \psi & \leftrightarrow & \left[\phi \land \left[\phi \right]^-_A \phi' \right]^-_A  \psi \label{form1} \\
\left[\phi \right]^+_A \left[\phi' \right]^+_A \psi & \leftrightarrow & \left[\phi \land \left[ \phi \right]^+_A \phi' \right]^+_A  \psi  \label{form2}
\end{eqnarray}
\end{proposition}

On the other hand, formulas~(\ref{eq1})-(\ref{form2}) do not hold if
the coalition appearing in the announcement operator is different from
the coalition appearing in the epistemic operator.

Given that operators $[\phi]^-_A$ and $[\phi]^+_A$ are not reducible,
it is of interest to investigate what kind of modalities they are,
specifically what modal principles their semantics validates. First, it
is easy to see that both axiom {\bf K} and rule {\bf Nec} of
necessitation are valid:
\begin{eqnarray*}
[\phi]^-_A(\psi \to \psi')  \to  ([\phi]^-_A \psi \to [\phi]^-_A \psi') \\
\psi \Rightarrow [\phi]^-_A \psi 
\end{eqnarray*}
and the same validities hold for operator $\left[ \phi \right]_A^+$.

On the other hand, all axioms {\bf T}, {\bf 4} and {\bf B} fail.  As
regards {\bf T}, if $\phi$ is false, then $[\phi]_a \psi$ holds for
any formula $\psi$, but it does not follow that $\psi$ holds whenever
it is false itself.  As to {\bf 4}, notice that in the muddy children
puzzle a child not stepping forward is tantamount to globally stating
that she does not know whether she is muddy, or $[\textit{no\_st}]
:= [\bigwedge_{i \in Ag} \neg Kw_i m_i]^+_{Ag}$. Hence, after the
father's announcement, in state $(1,1,0)$ we have that no child knows
whether she is muddy after the first round, that is, $(1,1,0) \models
[no\_st] \bigwedge_{i \in Ag} \neg Kw_i m_i$.  However, at the
second round all muddy children know that they are muddy: $(1,1,0)
\models [no\_st] [no\_st] \bigwedge_{i \in Ag} (m_i \to Kw_i
m_i)$.  In particular, $(1,1,0) \not \models [no\_st][no\_st]
\bigwedge_{i \in Ag} \neg Kw_i m_i$, thus invalidating {\bf 4}.
As regards {\bf B}, we can show that it fails by considering Moore's formula $p \land \neg K_a p$ and a pointed model $(\M, w)$ such that 
$(\M, w) \models p \land \neg K_a p$ but $(\M, w) \not \models [p]_a
\langle p \rangle_a (p \land \neg K_a p)$.

\section{GLAL, PAL, and Epistemic Logic}

The main result in this section is that GLAL, differently from PAL, is
not reducible to epistemic logic, but rather strictly more expressive
than both.
%
In Section~\ref{comparison} we proved that 
GLAL is at least as expressive as public
announcement logic. 
Next we prove that GLAL is strictly more expressive,
in the sense that some formulas in GLAL are not equivalent to
any epistemic formula. Since epistemic logic and PAL are equally
expressive \cite{baltagetal:1998}, it immediately follows that GLAL is
strictly more expressive than PAL as well.
%
\begin{theorem} \label{theor1}
GLAL is strictly more expressive than epistemic logic with common knowledge.
\end{theorem}
%
\begin{proofsk}
We prove this result by providing two modally bisimilar models, that
therefore satisfy the same epistemic formulas, but satisfy different
formulas in GLAL.
Consider models $\N$ and 
$\N'$
in Example~\ref{ex:bisim} 
and define a relation $B$ such that $B(w_i, w'_i)$ and $B(w_i, v'_i)$, for $i \in \{1, 2\}$.
It is easy to check that the $B$ is a modal bisimulation between
pointed models $(\N, w_1)$ and $(\N', w'_1)$. In particular, the same
epistemic formulas are satisfied at states $w_1$ and $w'_1$.
However, as we noticed in Example~\ref{ex:bisim},
for $\phi ::= [bit = 0]_r K_e Kw_r (bit = 0)$, we can check
$(\N, w_1) \models \phi$; while $(\N', w'_1) \not \models \phi$.  In
particular, there is no epistemic formula in $\L_{el}$ that is
equivalent to $\phi$.
\end{proofsk}


By Theorem~\ref{theor1} and the equi-expressivity of epistemic logic
and PAL \cite{baltagetal:1998}, we immediately obtain the following corollary.
\begin{corollary} 
GLAL is strictly more expressive than PAL.
\end{corollary}

By Example~\ref{ex:bisim} and the proof of Theorem~\ref{theor1} not even announcements to single
agents are reducible to epistemic formulas. 
Also, the same proof points out that a more robust notion of
bisimulation is needed to preserve formulas in GLAL. We explore such a
notion in the next section.

 \label{expressivity}

\section{Bisimulations} \label{bisimulations}

In this section we investigate a stronger notion of bisimulation
capable of preserving the meaning of GLAL formulas as well.  Firstly,
for any set $A \subseteq \Agents$ of agents, and model $\M = \langle
W, \{R_a\}_{a \in Ag}, V \rangle$, we define $R_A(w,v)$ as: $R_a(w,v)$
iff $a \in A$, that is, $R_A(w,v)$ holds iff $R_a(w,v)$ holds for
exactly the agents $a \in A$.
\begin{definition}[$\pm$-Simulation] \label{def:bisim}
 Given models $\M$ and $\M'$,
a {\em $\pm$-simulation} is a relation ${\mathbf S} \subseteq
W \times W'$ such that ${\mathbf S}(w, w')$ implies
\begin{enumerate}
\item[] {\bf Atoms \ }  $w \in V(p)$ iff $w' \in V'(p)$, for every $p \in AP$;
\item[] {\bf Forth \ } for every $A \subseteq \Agents$ and $v \in W$,  if $R_A(w,v)$ then for some $v' \in W'$, $R'_A(w',v')$ and ${\mathbf S}(v,v')$;
\item[] {\bf Reach \ } if ${\mathbf S}(v,v')$ then for every $a \in Ag$, $R_a(w,v)$ iff  $R'_a(w',v')$.
\end{enumerate}
\end{definition}


Besides conditions {\bf Atoms} and {\bf Forth},  {\bf Reach} is required to preserve the interpretation of formulas when refinements are considered.   
A {\em $\pm$-bisimulation} is a relation ${\mathbf B} \subseteq
W \times W'$ such that both ${\mathbf B}$ and ${\mathbf B}^{-1} = \{
(w',w) \mid {\mathbf B}(w,w') \}$ are simulations.
Two worlds $w$, $w'$ are $\pm$-bisimilar, or $w \rightleftharpoons w'$, iff
${\mathbf B}(w,w')$ holds for some $\pm$-bisimulation ${\mathbf
B}$. Further, two models $\M$ and $\M'$ are $\pm$-bisimilar, or
$\M \rightleftharpoons \M'$, iff
\myi for every $w \in \M$, $w \rightleftharpoons w'$ for some $w' \in \M'$; and \myii for every $w' \in \M'$, $w \rightleftharpoons w'$ for some $w \in \M$. 
We write $(\M,w) \rightleftharpoons (\M', w')$ to state that $w$
and $w'$ are $\pm$-bisimilar in $\pm$-bisimilar $\M$, $\M'$.

Clearly, by Def.~\ref{def:bisim}  $\pm$-bisimilarity
implies $m$-bisimilarity.
However, the opposite implication
does not hold. Notably, the $m$-bisimilar models $\N$ and $\N'$ in
Example~\ref{ex:bisim} are not $\pm$-bisimilar according to
Def.~\ref{def:bisim}.
We now show that global and local announcements preserve bisimilarity.
\begin{theorem}\label{thm:bisimpreservation}
Suppose that $(\M,s) \rightleftharpoons (\M', s')$. 
Then, for every formula $\psi$ in GLAL 
\begin{eqnarray*}
 (\M, s) \models \psi & \text{iff} & (\M' , s') \models \psi 
\end{eqnarray*}
In particular, if $\psi = [\theta]^-_{A} \theta'$ (resp.~$\psi =
[\theta]^+_{A} \theta'$) then 
${\M}^-_{(s,\theta,A)} \rightleftharpoons \M'^-_{(s',\theta, A)}$ and  ${\M}^+_{(s,\theta,A)} \rightleftharpoons \M'^+_{(s',\theta,A)}$.
\end{theorem}

\paragraph*{Discussion: Distributed knowledge}
There is a close relation between our notion of $\pm$-bisimulation and
the notion known in the literature as {\em collective
bisimulation} \cite{DBLP:journals/jancl/Roelofsen07,DBLP:journals/synthese/WangA13}. The
(only) difference is that in Def.~\ref{def:bisim} we require $R_A$ to
be exact ($(w,v)\in R_A$ iff: $(w,v) \in R_a$ for $a \in A$ and
$(w,v) \not\in R_a$ for $a \not\in A$) whereas in collective
bisimulation $R_A$ is taken inclusively ($(w,v)\in R_A$ iff:
$(w,v) \in R_a$ for $a \in A$). Clearly, $\pm$-bisimilar models are
also collectively bisimilar, so $\pm$-bisimilarity also implies
equivalence in the logic of distributed knowledge. However, the other
way round may not hold.
As an example, consider again models $\N$ and $\N'$ in
Example~\ref{ex:bisim}. These are collectively bisimilar, as well as
$m$-bisimilar. However, they are not $\pm$-bisimilar.  In particular,
it is easy to see that $(\N, w_1) \models [bit = 0]_r K_e D_{r,e} (bit
= 0)$; whereas $(\N', w'_1) \not \models [bit = 0]_r K_e D_{r,e} (bit
= 0)$. As a consequence, we obtain the following result:
\begin{theorem}
Epistemic logic with distributed
knowledge is not as expressive as GLAL.  
\end{theorem}

This begs the question as to how logics with distributed knowledge
relate to GLAL and, for example, what their relative expressivity is.
For instance, we do not know whether GLAL is strictly more expressive
or the two logics are uncomparable.

\section{Model Checking and Satisfaction} \label{modelcheck}

As part of the analysis of the formal properties of GLAL, we investigate the model checking and satisfiability problems, defined as follows.
\begin{definition}[Model Checking and Satisfiability]
\begin{itemize}
\item[]
\item {\bf Model Checking Problem}:
 Given a finite model $\M$, state $w$ in $\M$, and formula $\phi$ in GLAL, determine
 whether $(\M, w) \models \phi$.
\item {\bf Satisfiability Problem}:
 Given a formula $\phi$ in GLAL, determine whether $(\M, w) \models
  \phi$ for some model $\M$ and state $w$ in $\M$.
\end{itemize}
\end{definition}

The model checking problem is tantamount to determine whether $w \in
[\![\phi]\!]_{\M}$, hence it depends crucially on the complexity of
computing the satisfaction set $[\![\phi]\!]_{\M}$, as membership is
supposed to be computable in polynomial time.
\begin{theorem} \label{mcheck}
  The model checking problem for GLAL is PTIME-complete.
\end{theorem}

%

A consequence of Theorem~\ref{mcheck} is that model checking GLAL is
no more computationally complex than the verification of epistemic
logic. Hence, the enhanced expressivity of GLAL comes at no extra computational cost from a verificational perspective.

\begin{theorem} \label{sat}
  The satisfiability problem for GLAL is decidable.
\end{theorem}
In Theorem~\ref{sat} the decidability of GLAL is proved similarly to
the decidability of PAL. The difference is merely in the amount of
transitions due to announcements. Whereas in PAL announcements
$[\phi]$ are functional, in GLAL announcements $[\phi]^+_A$ and
$[\phi]^-_A$ are branching (i.e., for each $A \subseteq Ag$, and for
global and local announcements, we may need different transitions). In
that sense the decidability proof is more akin to that of action model
logic, wherein non-deterministic actions also cause branching.

\section{Conclusions} \label{conc}

We introduced a unified account to formalise both global and local
announcements in GLAL, a strictly more expressive extension of public
announcement logic.  The key feature of the semantics of GLAL is that
the refinement of the indistinguishability relations is defined in the
same way for public and private announcement, i.e., as the restriction
of the equivalence classes to the worlds satisfying (or not) the given
announcement. The crucial difference between global and local
announcements is the domain of application of such updates: worlds
accessible in one step
or all worlds epistemically reachable, respectively.
%
In Example~\ref{ex1} and \ref{ex:bisim} we showed how these formal
notions capture our intuitions about global and local announcements.

In future work we envisage several extensions.  Firstly, since
differently from PAL, announcements are not necessarily broadcast to
all agents (so that only one such announcement can be broadcast at
each given time), we can envisage global and local announcements
communicated simultaneously and introduce formulas $([\phi]_A \circ
[\phi']_B) \psi$ with the intended meaning that if $\phi$ is
(truthfully) announced to coalition $A$ and {\em simultaneously}
$\phi'$ is announced to coalition $B$, then $\psi$ holds.  This is of
interest to model synchronous communication. Particular care is to be
taken in defining the semantics of operator $[\phi]_A \circ [\phi']_B$
whenever the intersection of coalitions $A$ and $B$ is non-empty.

Secondly, as the receiver of the announcement can be a subset $A
\subseteq Ag$ of the set of all agents, we can think that the announcement
originates from some other coalition $B$ and introduce GLAL operators
$[\phi]_{B,A}$ indexed to both $A$ and $B$.
Such an extension would provide a finer-grained analysis of scenarios
such as communication and security protocols.


\subsection*{Acknowledgements} We thank the TARK reviewers for the comments. The research described in this paper was 
supported by the French ANR JCJC Project SVeDaS
(ANR-16-CE40-0021) and the ERC project EPS 313360. Hans van Ditmarsch is also affiliated to IMSc, Chennai.

\bibliographystyle{eptcs}
\bibliography{tarkbib}  

\appendix \label{app}

\section*{Appendix with Proofs}

\begin{proof} {\bf (Proposition~\ref{prop10})}
 The only non-trivial case is for $\psi = [\phi] \phi'$. In
 particular, we show that for every $w \in W$, refinement
 $\M^+_{(w, \phi, Ag)}$ satisfies the same formulas in PAL as
 refinement $\M_{\phi}$.
 The key remark here is that worlds that are not reachable from $w$
 via relation $R^C_{Ag}$ do not account for the satisfaction of PAL
 formulas at $w$. Specifically, in refinement $\M^+_{(w, \phi, Ag)}$,
 any state $w'$ is reachable from $w$ via $R^{+C}_{Ag}$ iff $w'$ is
 reachable from $w$ in $\M_{\phi}$. Also, in both refinements the
 indistinguishability relations and assignments are restricted to
 $\denot{\phi}_{\M}$. As a result, the two models satisfy the same
 $[\phi]$-formulas at $w$.
\end{proof}

\begin{proof} {\bf (Proposition~\ref{lem1})}
We prove~(\ref{ref1}) for a propositional formula $\phi$.
Suppose that
$(\M, w) \models \phi$ but $(\M^-_{(w, \phi, A)}, w) \not \models
E_A \phi$ to obtain a contradiction, that is, $(\M^-_{(w, \phi, A)}, w')
\not \models \phi$ for some $a \in A$ and $w' \in R^-_a(w)$.  In particular, this
means that $(\M, w') \not \models \phi$, as $\phi$ is propositional.
Hence, $w' \neq w$ (as $\phi$ is true in $w$) and $w' \in
R_a(w) \supseteq R^-_a(w)$.
But then $w' \notin R_a(w) \cap [\![ \phi ]\!]$,
against the hypothesis that $R^-_a(w,w')$. Therefore, it is the
case that $(\M^-_{(w,\phi,A)}, w) \models E_A \phi$.
The proof for~(\ref{ref2}) follows a similar line. 
\end{proof}

\begin{proof} {\bf (Proposition~\ref{prop20})}
%
%
We prove (\ref{form1}). Suppose that $(\M,
w) \models \left[\phi \right]^-_A \left[\phi' \right]^-_A \psi$, that
is, if $(\M, w) \models \phi$ and $(\M^-_{(w, \phi, A)},
w) \models \phi'$, then $((\M^-_{(w, \phi, A)})^-_{(w, \phi', A)} ,
w) \models \psi$.  We have to show that this is equivalent to $(\M,
w) \models \left[\phi \land \left[\phi \right]^-_A \phi' \right]^-_A \psi$,
that is, if $(\M,
w) \models \phi$ and $(\M^-_{(w, \phi, A)},
w) \models  \phi'$, then
$(\M^-_{(w, \phi \land \left[\phi \right]^-_A \psi', A)},
w) \models \psi$.
Hence, it is enough to prove that $((\M^-_{(w, \phi, A)})^-_{(w, \phi', A)} ,
w) \models \psi$ iff $(\M^-_{(w, \phi \land \left[\phi \right]^-_A \psi', A)},
w) \models \psi$. In particular, refinements $(\M^-_{(w, \phi, A)})^-_{(w, \phi', A)}$ and $\M^-_{(w, \phi \land \left[\phi \right]^-_A \psi', A)}$ are identical.
To see this we remark that in refinement $(\M^-_{(w, \phi, A)})^-_{(w, \phi', A)}$, for
every $a \in A$,
{\footnotesize
\[
R^-_a(v) =  
  \begin{cases} R_a(v) \cap [\![\phi]\!]_{\M} \cap [\![\phi']\!]_{\M_{(w, \phi, A)}}\\
\hspace{1.5cm} \mbox{if } v \in R_a(w) \cap [\![\phi]\!]_{\M}  \cap [\![\phi']\!]_{\M_{(w, \phi, A)}}\\
 R_a(v) \cap [\![\phi]\!]_{\M}  \cap [\![\neg \phi']\!]_{\M_{(w, \phi, A)}}\\
\hspace{1.5cm} \mbox{if } v \in R_a(w) \cap [\![\phi]\!]_{\M}  \cap [\![\neg \phi']\!]_{\M_{(w, \phi, A)}}\\
R_a(v) \cap [\![\neg \phi]\!]_{\M}  \cap [\![\phi']\!]_{\M_{(w, \phi, A)}}\\
\hspace{1.5cm}  \mbox{if } v \in R_a(w) \cap [\![\neg \phi]\!]_{\M}  \cap [\![\phi']\!]_{\M_{(w, \phi, A)}}\\
R_a(v) \cap [\![\neg \phi]\!]_{\M}  \cap [\![\neg \phi']\!]_{\M_{(w, \phi, A)}}\\
\hspace{1.5cm} \mbox{if } v \in R_a(w) \cap [\![\neg \phi]\!]_{\M}  \cap [\![ \neg \phi']\!]_{\M_{(w, \phi, A)}}\\
  R_a(v)\\
\hspace{1.5cm} \mbox{otherwise} 
  \end{cases}
  \]
}
which is tantamount to the following in model $\M^-_{(w, \phi \land \left[\phi \right]^-_A \phi', A)}$:
{\scriptsize
\[
R^-_a(v) =  
  \begin{cases} R_a(v) \cap [\![\phi \land [\phi]_A \phi' ]\!]_{\M} & \mbox{if } v \in R_a(w) \cap [\![\phi \land [\phi]_A \phi' ]\!]_{\M}\\
 R_a(v) \cap [\![\neg (\phi \land [\phi]_A \phi')]\!]_{\M} & \mbox{if } v \in R_a(w) \cap [\![\neg (\phi \land [\phi]_A \phi')]\!]_{\M}\\
  R_a(v) & \mbox{otherwise} 
  \end{cases}
  \]
}
Hence, the two models are identical and (\ref{form1}) holds.
\end{proof}

\begin{proof} {\bf (Theorem \ref{thm:bisimpreservation})}
The proof is by induction on the structure of formula $\psi$.  The
base case for atomic propositions and the inductive cases for
propositional connectives and epistemic modalities are immediate, as
$\pm$-bisimulations
are 
$m$-bisimulations in particular.
So suppose that $\psi = [\theta]^-_A \theta'$, and assume the
induction hypothesis for both $\theta$ and $\theta'$. We can guarantee
this by assuming that the complexity of $\theta$
is lower than that of $\theta'$. This, in turn, can be achieved for
instance by allowing only update formulas of the form
$[\theta]^-_A((\theta' \land \theta) \lor (\theta' \land \neg \theta))$.

Next, we consider the local refinements ${\M}^-_{(s,\theta,A)}$ and
 ${\M}'^-_{(s',\theta,A)}$, and prove {\bf Atoms}, {\bf Forth} and
{\bf Reach} from left to right, the other direction
being symmetric.
We assume a bisimulation ${\mathbf B} \subseteq W \times W'$ and
show that ${\mathbf B}$ is a bisimulation between
${\M}^-_{(s,\theta,A)}$ and ${\M'}^-_{(s',\theta,A)}$ as well, such that ${\mathbf B}(s,s')$ holds. 
Firstly, for every $w \in W^- = W$ and $w' \in W'^- = W'$, if
${\mathbf B}(w,w')$ then $w \in V^-(p)$ iff $w \in V(p)$ iff $w' \in
V'(p)$ iff $w' \in V'^-(p)$, which shows that {\bf Atoms} is
satisfied.

As regards {\bf Forth}, let us write $R^E_A(s)$ for $\bigcup_{a \in A}
R_a(s)$.  Then, we distinguish the case whether $w \in R^E_A(s)$ or $w
\notin R^E_A(s)$. First, if ${\mathbf B}(w,w')$ and $w \not \in
R^E_A(s)$, then for every set $B \subseteq \Agents$ of agents and $v
\in W$, $R^{-}_B(w,v)$ iff $R_B(w,v)$.  Moreover, since ${\mathbf
  B}(s,s')$, $w' \not \in R'^E_A(s')$ by {\bf Reach}, and therefore
for some $v' \in W'$, ${\mathbf B}(v,v')$ and $R'_B(w',v')$ implies
$R'^-_B(w',v')$.

Now consider $w \in R^E_A(s)$ (and therefore $w' \in R^E_A(s')$ by
{\bf Reach}). We distinguish two cases: \myi $(\M,w) \models \theta$ iff
$(M,v) \models \theta$ (i.e., $(\M,w)$ and $(\M,v)$ agree on $\theta$). By
the definition of refinement, we then have that $R^-_B(w,v)$ iff
$R_B(w,v)$. Since ${\mathbf B}$ is a bisimulation, we then also have
$R'_B(w',v')$ for some $v' \in W'$. Moreover, by induction hypothesis,
$(\M',w')$ and $(\M', v')$ agree on $\theta$ as well. So we have
$R'^-_B(w',v')$.

Now assume \myii $(\M,w) \models \theta$ iff
$(\M,v) \not \models \theta$ (i.e., $(\M,w)$ and $(\M,v)$ disagree on
$\theta$). Suppose that $R^-_B(w,v)$. By the definition of refinement,
we have $R_B(w,v)$, where $B$ is a set for which $C \subseteq
B \subseteq (A \cup C)$, that is, all agents in $A \cap B$ have
learned that $w$ and $v$ are different. But then we also have
$R'_B(w',v')$ for some $v' \in W'$, and again by the definition of
refinement, $R'^-_B(w',v')$.

Finally, condition {\bf Reach} holds for ${\M}^-_{(s,\theta,A)}$ and
${\M'}^-_{(s',\theta,A)}$, as if ${\mathbf B}(w,w')$, ${\mathbf
  B}(v,v')$, and $w \in R^E_A(s)$, then for every $b \in Ag$,
$R_b(w,v)$ implies that either $b \notin A$ or $w$ and $v$ agree on
the interpretation of $\theta$ in $\M$. By induction hypothesis $w'$
and $v'$ agree as well on $\theta$, and in particular $R'_b(w',v')$.

For the global refinement ${\M}^+_{(s,\psi,A)}$,
the proof is similar, but instead of equivalence class $R^E_A(w)$, we
consider $R^C_A(w)$.
\begin{figure}[ht]
\begin{center}
\includegraphics[width=.45\textwidth]{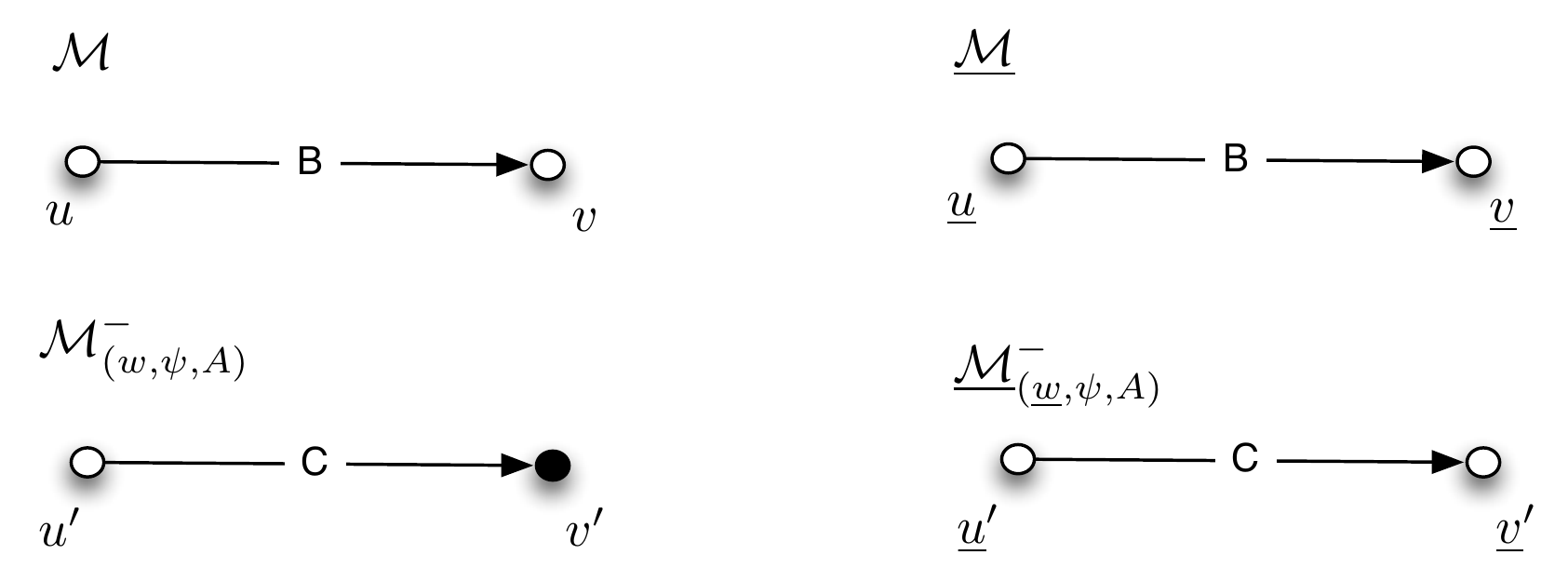} 
\end{center}
\caption{Scheme for the proof of Theorem~\ref{thm:bisimpreservation}}\label{fig:fourmodelsb}
\end{figure}


To conclude the proof, for the case of announcement formulas $\psi =
[\theta]^-_A \theta'$, $(\M,w) \models \psi$ iff
$(\M,w) \models \theta$ implies $(\M^-_{(w,\theta,A)}, w) \models
\theta'$).  If $(\M,w) \not\models \theta$, then also $(\M',w')
\not\models\theta$, so $(\M',w') \models \psi$. Then,
suppose that $(\M,w) \models \theta$.  We then have
$(\M^-_{(w,\theta,A)}, w) \models \theta'$ and, since
$(\M^-_{(w,\theta,A)}, w) \leftrightharpoons
(\M'^-_{(w',\theta,A)},w')$ by the proof above, by
the induction hypothesis we obtain $(\M'^-_{(w',\theta,A)}, w')
\models \theta'$, and therefore $(\M', w') \models \varphi$. The case
for $\psi = [\theta]^+_A \theta'$ is similar.
\end{proof}

\begin{proof} {\bf (Theorem~\ref{mcheck})}
 The PTIME-hardness of model checking GLAL follows from the
 PTIME-hardness of model checking epistemic logic and the fact that
 the latter is subsumed by the former \cite{HoekM95}.

As regards membership in PTIME, epistemic formulas can be checked in
polynomial time. Further, model refinements $\M^+_{(w, \psi, A)}$ and
$\M^-_{(w, \psi, A)}$, for world $w$, formula $\psi$, and coalition
$A$, can be computed in polynomial time in the size $|\M|$ of the
model given as $|W| + \sum_{a \in Ag}|R_a|$. By combining the two
procedures we can obtain an algorithm that runs in polynomial time.
\end{proof}

\begin{proofsk} {\bf (Theorem~\ref{sat})}
To prove the decidability of GLAL we describe a decision procedure
inspired by \cite{halpernetal:1992}.  There, given an epistemic
formula $\phi \in \L_{el}$, the authors provide an algorithm to build
a (purely epistemic) model $M = (S, \pi, \K_1, \ldots, \K_n)$ in
exponential time, in which \myi each state $s \in S$ is a maximal
propositional tableau on the set $Sub^+_C(\phi)$ of subformulas of
$\phi$\footnote{In \cite{halpernetal:1992} $Sub(\phi)$ is the set of
  all subsets of $\phi$; $Sub_C(\phi)$ extends $Sub(\phi)$ with
  formulas $E(\psi \land C\psi)$ and $\psi \land C\psi$ for every
  subformula $\psi \in Sub(\phi)$; and $Sub^+_C(\phi)$ includes
  $Sub_C(\phi)$ and the negation of formulas in $Sub_C(\phi)$.}; \myii
$p \in \pi(s)$ iff $p \in s$; \myiii $\K_a(s,t)$ iff $\{ K_a \psi \mid
K_a \psi \in s\} = \{ K_a \psi \mid K_a \psi \in t\}$. In particular,
if $\phi$ if satisfiable (in some epistemic model), then $M$ is a
model for $\phi$.

Now, based on the above, we describe a procedure to build a (possibly
empty) model for a formula $\phi \in \L_{glal}$. Firstly, by
interpreting each announcement subformula in $\psi$ as a new
proposition, $\phi$ can be seen as a purely epistemic formula, so that
the procedure in \cite{halpernetal:1992} applies. Then, we obtain an
epistemic model $M$.  Further, starting with $M$ we build a GLAL model
for $\phi$. 
Specifically, given $M$ and copies $s'_1, \ldots, s'_n$ of states
$s_1, \ldots, s_n$ in $M$, consider a new model $M' = (S', \pi, \K'_1,
\ldots, \K'_n)$ such that \myi $S' = S \cup \{ s'_1, \ldots, s'_n \}$;
and \myii for $s, t \in S$, $\K'_a(s,t)$ iff $\K_a(s,t)$, while for $t
\in S$, $\K'_a(s',t)$ and $\K'_a(t,s')$ only if $\K_a(s,t)$ and
$\K_a(t,s)$.  Notice that all such $M'$ are $m$-bisimilar to the
original $M$. However, they are not necessarily $\pm$-bisimilar to
$M$, and therefore we need to consider all such $M'$ up to
$\pm$-bisimulation when checking for satisfiability.  In particular,
these are in finite number as $S$ is finite.

Finally, we define a tableau for GLAL as a tuple $T = \langle \{ M_i \}_{i
  \in I}, \{ R^-_{(s, \phi, A)} \}, \{ R^+_{(s, \psi, A)} \} \rangle$, for $s \in S, \psi \in Sub^+_C(\phi), A
  \subseteq Ag$, such that
\begin{itemize}
\item all $M_i = \langle S, \pi, \K_1, \ldots, \K_n \rangle$ are
  epistemic models, all defined on a single
  set $S$ of states, as described above;
\item each $R^-_{(s, \psi, A)}$ is a relation on epistemic models
  such that $R^-_{(s, \psi, A)}(M, M')$ and $\{ \psi, [ \psi
  ]_A^- \chi \} \subseteq \pi(s)$ imply
\begin{itemize}
\item $\chi \in \pi'(s)$;
\item for every $t \in \K_a(s)$, for every $b \notin A$, 
$\{K_b \theta \mid K_b \theta \in \pi'(t) \} = \{K_b \theta \mid K_b \theta \in \pi(t) \}$, while for every $a \in A$,
$\{K_a \theta \mid K_a \theta \in \pi'(t) \} = \{K_a \theta \mid K_a \theta \in \pi(t) \} \cup \{ K_a \theta \mid K_a \theta \in Sub^{+}_C (K_a \psi) \}$ iff $\psi \in \pi(t)$ and $\{K_a \theta \mid K_a \theta \in \pi'(t) \} = \{K_a \theta \mid K_a \theta \in \pi(t) \} \cup \{ K_a \theta \mid K_a \theta \in Sub^{+}_C (K_a \neg \psi) \}$
 iff $\neg \psi \in \pi(t)$;
\item for every $t \notin \K_a(s)$, $\pi'(t) = \pi(t)$.
\end{itemize}
\item each $R^+_{(s, \psi, A)}$ is a relation on epistemic structures
  such that $R^+_{(s, \psi, A)}(M, M')$ and $\{ \psi, [ \psi ]_A^+
  \chi \} \subseteq \pi(s)$ imply
\begin{itemize}
\item $\chi \in \pi'(s)$;
\item for every $t \in (\bigcup_{a \in A} \K_a)^*(s)$, for every $b \notin A$, 
$\{K_b \theta \mid K_b \theta \in \pi'(t) \} = \{K_b \theta \mid K_b \theta \in \pi(t) \}$, while for every $a \in A$,
$\{K_a \theta \mid K_a \theta \in \pi'(t) \} = \{K_a \theta \mid K_a \theta \in \pi(t) \} \cup \{ K_a \theta \mid K_a \theta \in Sub^{+}_C (C_A \psi) \}$ iff $\psi \in \pi(t)$ and $\{K_a \theta \mid K_a \theta \in \pi'(t) \} = \{K_a \theta \mid K_a \theta \in \pi(t) \} \cup \{ K_a \theta \mid K_a \theta \in Sub^{+}_C (C_A \neg \psi) \}$
 iff $\neg \psi \in \pi(t)$;
\item for every $t \notin (\bigcup_{a \in A} \K_a)^*(s)$, $\pi'(t) =
  \pi(t)$.
\end{itemize}
\end {itemize}

Clearly, if any tableau $T$ for a formula $\phi$ as above is
non-empty, then $\phi$ is satisfiable.  In particular, we can prove
that for every $\psi \in Sub^+_C(\phi)$, $(M, s) \models \psi$ iff
$\psi \in \pi(s)$, for $s \in M$.

Furthermore, the soundness of the procedure follows by the soundness
of the three steps separately. Firstly, $\phi$ is satisfiable as an
epistemic formula iff then there exists at least one epistemic model
$M$ defined as above. Given $M$, there are finitely many $M'$ non
$\pm$-bisimilar to $M$. Then, we attempt to build a GLAL tableau $T$ starting in each of these $M'$.
In particular, $T$ is not empty iff $\phi$ is satisfiable.
\end{proofsk}

\end{document}